\documentclass{llncs}
\usepackage{etex}
\usepackage{makeidx}  
\usepackage{amssymb,upgreek,nicefrac,amsmath}
\usepackage{sgame}
\usepackage[mathscr]{eucal}
\usepackage{color}
\usepackage[dvips]{epsfig}
\usepackage[dvips]{graphicx}
\usepackage[section]{placeins}
\usepackage{pst-node}
\usepackage{dsfont}
\usepackage{mathrsfs}
\usepackage{sgame}
\usepackage{float}
\usepackage[T1]{fontenc}
\usepackage{aurical}
\usepackage{tikz,pgfplots,verbatim}
\usetikzlibrary{spy,arrows,backgrounds,plotmarks,shapes,snakes}
\usepackage[square,sort,comma,numbers]{natbib}

\definecolor{Yellow}{rgb}{1, 1, 0}
\definecolor{VeryLightGray}{gray}{.90}
\definecolor{LightGray}{gray}{.7}
\definecolor{Gray}{gray}{.50}
\definecolor{DarkGray}{gray}{.3}
\definecolor{VeryDarkGray}{gray}{.10}

\newcommand{\tr}{^{\mathrm T}}

\newcommand{\magn}[1]{\left\vert #1 \right\vert}

\newcommand{\bitem}{\item[$\bullet$]}

\newcommand{\df}{\doteq}

\newcommand{\RM}{$\mathsf{RM}$}

\newcommand{\SIMPLEX}[1]{\Delta\left(#1\right)}
\newcommand{\SPROFILE}{\mathbf{\Delta}}

\newcommand{\RAND}[2]{{\sf rand}_{#1}\left[#2\right]}

\begin{document}

\title{Efficient Dynamic Pinning of Parallelized Applications by Distributed Reinforcement Learning\thanks{This work has been partially supported by the European Union grant EU H2020-ICT-2014-1 project RePhrase (No. 644235).}
}

\titlerunning{Efficient Dynamic Pinning for Parallelized Applications}        

\author{Georgios C. Chasparis \and Michael Rossbory
}


\institute{Software Competence Center Hagenberg GmbH, Softwarepark 21, A-4232 Hagenberg, Austria \\
\email{\{georgios.chasparis,michael.rossbory\}@scch.at} 
}



\maketitle

\begin{abstract}
This paper introduces a resource allocation framework specifically tailored for addressing the problem of dynamic placement (or pinning) of parallelized applications to processing units. Under the proposed setup each thread of the parallelized application constitutes an independent decision maker (or agent), which (based on its own prior performance measurements and its own prior CPU-affinities) decides on which processing unit to run next. Decisions are updated recursively for each thread by a resource manager/scheduler which runs in parallel to the application's threads and periodically records their performances and assigns to them new CPU affinities. For updating the CPU-affinities, the scheduler uses a distributed reinforcement-learning algorithm, each branch of which is responsible for assigning a new placement strategy to each thread. According to this algorithm, prior allocations are going to be reinforced in the future proportionally to their prior performance. The proposed resource allocation framework is flexible enough to address alternative optimization criteria, such as maximum average processing speed and minimum speed variance among threads. We demonstrate analytically that convergence to locally-optimal placements is achieved asymptotically. Finally, we validate these results through experiments in Linux platforms. 
\end{abstract}

\section{Introduction}
Resource allocation has become an indispensable part of the design of any engineering system that consumes resources, such as electricity power in home energy management \cite{DeAngelis13}, access bandwidth and battery life in wireless communications \cite{Inaltekin05}, computing bandwidth under certain QoS requirements \cite{Bini11}, computing bandwidth for time-sensitive applications \cite{ChasparisMaggio16},  computing bandwidth and memory in parallelized applications \cite{Brecht93}. 

When resource allocation is performed online and the number, arrival and departure times of the tasks are not known a priori (as in the case of CPU bandwidth allocation), the role of a resource manager (\RM) is to guarantee an \emph{efficient} operation of all tasks by appropriately distributing resources. However, guaranteeing efficiency through the adjustment of resources requires the formulation of a centralized optimization problem (e.g., mixed-integer linear programming formulations \cite{Bini11}), which further requires information about the specifics of each task (i.e., application details). Such information may not be available to neither the \RM\ nor the task itself.

Given the difficulties involved in the formulation of centralized optimization problems, not to mention their computational complexity, feedback from the running tasks in the form of performance measurements may provide valuable information for the establishment of efficient allocations. Such (feedback-based) techniques have recently considered in several scientific domains, such as in the case of application parallelization (where information about the memory access patterns or affinity between threads and data are used in the form of scheduling hints)  \cite{Broquedis10}, or in the case of allocating virtual processors to time-sensitive applications \cite{ChasparisMaggio16}.

To tackle the issues of centralized optimization techniques, resource allocation problems have also been addressed through distributed or game-theoretic optimization schemes. The main goal of such approaches is to address a centralized (\emph{global}) objective for resource allocation through agent-based (\emph{local}) objectives, where, for instance, agents may represent the tasks to be allocated. Examples include the cooperative game formulation for allocating bandwidth in grid computing \cite{Sub08}, the non-cooperative game formulation in the problem of medium access protocols in communications \cite{Tembine09} or for allocating resources in cloud computing \cite{Wei10}. The main advantage of distributing the decision-making process is the considerable reduction in computational complexity (a group of $N$ tasks can be allocated to $m$ resources with $m^N$ possible ways, while a single task may be allocated with only $m$ possible ways). This further allows for the development of online selection rules where tasks/agents make decisions often using current observations of their \emph{own} performance.

This paper proposes a distributed learning scheme specifically tailored for addressing the problem of dynamically assigning/pinning threads of a parallelized application to the available processing units. Prior work has demonstrated the importance of thread-to-core bindings in the overall performance of a parallelized application. For example, \cite{Klug11} describes a tool that checks the performance of each of the available thread-to-core bindings and searches an optimal placement. Unfortunately, the \emph{exhaustive-search} type of optimization that is implemented may prohibit runtime implementation. Reference \cite{Broquedis10} combines the problem of thread scheduling with scheduling hints related to thread-memory affinity issues. These hints are able to accommodate load distribution given information for the application structure and the hardware topology. The HWLOC library is used to perform the topology discovery which builds a hierarchical architecture composed of hardware objects (NUMA nodes, sockets, caches, cores, etc.), and the BubbleSched library \cite{Thibault07} is used to implement scheduling policies. 
A similar scheduling policy is also implemented by \cite{Olivier11}.

This form of scheduling strategies exhibits several disadvantages when dealing with dynamic environments (e.g., varying amount of available resources). In particular, retrieving the exact affinity relations during runtime may be an issue due to the involved information complexity. Furthermore, in the presence of other applications running on the same platform, the above methodologies will fail to identify irregular application behavior and react promptly to such irregularities. Instead, in such dynamic environments, it is more appropriate to consider learning-based optimization techniques where the scheduling policy is being updated based on performance measurements from the running threads. Through such measurement- or learning-based scheme, we can a) \emph{reduce information complexity} (i.e., when dealing with a large number of possible thread/memory bindings) since only performance measurements need to be collected during runtime, and b) \emph{adapt to uncertain/irregular application behavior}.

To this end, this paper proposes a dynamic (algorithmic-based) scheme for optimally allocating threads of a parallelized application into a set of available CPU cores. The proposed methodology implements a distributed reinforcement learning algorithm (executed in parallel by a resource manager/scheduler), according to which each thread is considered an independent agent making decisions over its own CPU-affinities. The proposed algorithm requires minimum information exchange, that is only the performance measurements collected from each running thread. Furthermore, it exhibits adaptivity and robustness to possible irregularities in the behavior of a thread or to possible changes in the availability of resources. We analytically demonstrate that the reinforcement-learning scheme asymptotically learns a locally-optimal allocation, while it is flexible enough to accommodate several optimization criteria. We also demonstrate through experiments in a Linux platform that the proposed algorithm outperforms the scheduling strategies of the operating system with respect to completion time.

The paper is organized as follows. Section~\ref{sec:framework} describes the overall framework and objective. Section~\ref{sec:MultiagentFormulation} introduces the concept of multi-agent formulations and discusses their advantages. Section~\ref{sec:ReinforcementLearning} presents the proposed reinforcement-learning algorithm for dynamic placement of threads and Section~\ref{sec:ConvergenceAnalysis} presents its convergence analysis. Section \ref{sec:Experiments} presents experiments of the proposed algorithm in a Linux platform and comparison tests with the operating system's response. Finally, Section~\ref{sec:Conclusions} presents concluding remarks.

\textit{Notation:} 
\begin{itemize}
\bitem $|x|$ denotes the Euclidean norm of a vector $x\in\mathbb{R}^{n}$.
\bitem $\mathsf{dist}(x,A)$ denotes the minimum distance from a vector $x\in\mathbb{R}^{n}$ to a set $A\subset\mathbb{R}^{n}$, i.e., $\mathsf{dist}(x,A)\df\inf_{y\in{A}}|x-y|$.
\bitem $\mathcal{B}_{\delta}(A)$ denotes the $\delta$-neighborhood of a set $A\subset\mathbb{R}^{n}$, i.e., $\mathcal{B}_{\delta}(A)\df\{x\in\mathbb{R}^{n}:\mathsf{dist}(x,A)<\delta\}$.
\bitem For some finite set $A$, $\magn{A}$ denotes the cardinality of $A$.
\bitem The probability simplex of dimension $n$ is denoted $\SIMPLEX{n}$ and defined as 
$$\SIMPLEX{n}\df\Big\{x=(x_1,...,x_n)\in[0,1]^n:\sum_{i=1}^nx_i = 1\Big\}.$$
\bitem $\Pi_{\SIMPLEX{n}}[x]$ is the projection of a vector $x\in\mathbb{R}^{n}$ onto $\SIMPLEX{n}$.
\bitem $e_j\in\mathbb{R}^{n}$ denotes the unit vector whose $j$th entry is equal to 1 while all other entries are zero;
\bitem For a vector $\sigma\in\SIMPLEX{n}$, let $\RAND{\sigma}{a_1,...,a_n}$ denote the random selection of an element of the set $\{a_1,...,a_n\}$ according to the distribution $\sigma$;
\end{itemize}

\section{Problem Formulation \& Objective} 
\label{sec:framework}

\subsection{Framework}
\label{sec:RMapp}

We consider a resource allocation framework for addressing the problem of dynamic pinning of parallelized applications. In particular, we consider a number of threads $\mathcal{I}=\{1,2,...,n\}$ resulting from a parallelized application. These threads need to be pinned/scheduled for processing into a set of available CPU's $\mathcal{J}=\{1,2,...,m\}$ (not necessarily homogeneous). 

We denote the \emph{assignment} of a thread $i$ to the set of available CPU's by $\alpha_i \in \mathcal{A}_i \equiv \mathcal{J}$, i.e., $\alpha_i$ designates the number of the CPU where this thread is being assigned to. Let also $\alpha=\{\alpha_i\}_i$ denote the \emph{assignment profile}.

Responsible for the assignment of CPU's into the threads is the Resource Manager (\RM), which periodically checks the prior performance of each thread and makes a decision over their next CPU placements so that a (user-specified) objective is maximized. Throughout the paper, we will assume that:
\begin{itemize}
\item[(a)] The internal properties and details of the threads are not known to the \RM{}. Instead, the \RM\ may only have access to measurements related to their performance (e.g., their processing speed).
\item[(b)] Threads may not be idled or postponed. Instead, the goal of the \RM{} is to assign the \emph{currently} available resources to the \emph{currently} running threads. 
\item[(c)] Each thread may only be assigned to a single CPU core.
\end{itemize}

\begin{figure}[t!]
\centering
\begin{tikzpicture}[scale=3.5,
    axis/.style={thick, ->, >=stealth'},
    important line/.style={ultra thick},
    thin line/.style={thin},
    dashed line/.style={dashed, thin},
    arrow/.style={ultra thick, ->, >=stealth', shorten <=2pt, shorten >=2pt},
    every node/.style={color=black}
    ]
    \node[draw, rectangle, fill=LightGray] (RM) at (0.75,0) {\begin{minipage}{0.6\textwidth}\centering Resource Manager \textsf{(RM)} \\ 
    $\alpha^* = (\alpha_1^*,\alpha_2^*,...,\alpha_n^*) \doteq \arg\max_{\alpha\in\mathcal{A}}f(\alpha,w)$\end{minipage}};
   
    \node[draw, circle,fill=VeryLightGray] (T1) at (0.0,.7) {$T_1$};
    \node[draw, circle,fill=VeryLightGray] (T2) at (0.6,.7) {$T_2$};
    \node[draw, circle,fill=VeryLightGray] (T3) at (1.3,0.7) {$T_n$};
    \node (dots) at (1,0.7) {$\mathbf{\cdots}$};
    
    \node[draw, circle, fill=VeryLightGray] (CPU1) at (0,1.4) {CPU $1$};
    \node[draw, circle, fill=VeryLightGray] (CPU2) at (0.4,1.4) {CPU $2$};
    \node[draw, circle, fill=VeryLightGray] (CPU3) at (0.8,1.4) {CPU $3$};
    \node (dots) at (1.2,1.4) {$\mathbf{\cdots}$};
    \node[draw, circle, fill=VeryLightGray] (CPUm) at (1.6,1.4) {CPU $m$};
    
    \node[] (RMT1) at (.0,.1) {};
    \node[] (RMT2) at (.6,.1) {};
    \node[] (RMT3) at (1.3,.1) {};
    
    \draw[arrow] (RMT1) to [out=90,in=-90] (T1);\node at (0.1,0.45) {{\small $\alpha_1^*$}};
    \draw[arrow] (RMT2) to [out=90,in=-90] (T2);\node at (0.7,0.45) {{\small $\alpha_2^*$}};
    \draw[arrow] (RMT3) to [out=90,in=-90] (T3);\node at (1.4,0.45) {{\small $\alpha_n^*$}};
    
    \node[] (w1) at (.0,2) {};
    \draw[arrow] (w1) to [out=-90,in=90] (CPU1);\node at (0.1,1.9) {$w_1$};
    \node[] (w2) at (.4,2) {};
    \draw[arrow] (w2) to [out=-90,in=90] (CPU2);\node at (0.5,1.9) {$w_2$};
    \node[] (w3) at (.8,2) {};
    \draw[arrow] (w3) to [out=-90,in=90] (CPU3);\node at (0.9,1.9) {$w_3$};
    \node[] (wm) at (1.6,2) {};
    \draw[arrow] (wm) to [out=-90,in=90] (CPUm);\node at (1.7,1.9) {$w_m$};
    
    \node at (-0.1,0.86) {{\small $\alpha_1^*$}};
    \node at (0.48,0.86) {{\small $\alpha_2^*$}};
    \node at (1.18,0.86) {{\small $\alpha_n^*$}};
    
    
    \draw[arrow] (T1) to [out=90,in=-90] (CPU2);
    \draw[arrow] (T2) to [out=90,in=-90] (CPU1);
    \draw[arrow] (T3) to [out=90,in=-90] (CPU3);

\end{tikzpicture}
\caption{Schematic of \emph{static} resource allocation framework.}
\label{fig:StaticApproach}
\end{figure}
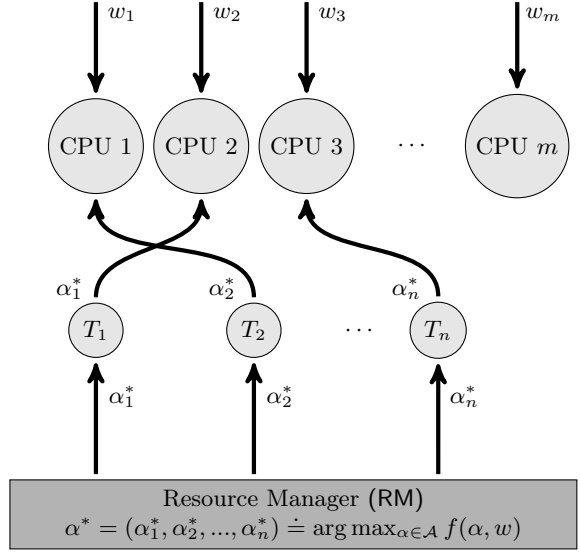

\subsection{Static optimization \& issues} 	\label{sec:StaticOptimization}

The selection of a centralized objective is open-ended. In the remainder of the paper, we will consider two main possibilities of a centralized objective in order to emphasize the flexibility of the introduced methodology to address alternative criteria. In the first case, the centralized objective will correspond to maximizing the average processing speed. In the second case, it will correspond to maximizing the average processing speed while maintaining a balance between the processing speeds of the running threads. 

Let $v_i=v_i(\alpha,w)$ denote the processing speed of thread $i$ which depends on both the overall assignment $\alpha$, as well as exogenous parameters aggregated within $w$. The exogenous parameters $w$ summarize, for example, the impact of other applications running on the same platform or other irregularities of the applications. Then, the previously mentioned centralized objectives may take on the following form:
\begin{eqnarray} \label{eq:CentralizedObjective}
\max_{\alpha\in\mathcal{A}} & f(\alpha,w), 
\end{eqnarray}
where 
\begin{enumerate}
\item[(O1)] $f(\alpha,w) \doteq \sum_{i=1}^{n} v_i/n$, corresponds to the average processing speed of all threads;
\item[(O2)] $f(\alpha,w) \doteq \sum_{i=1}^{n} [v_i - \gamma(v_i-\sum_{j\in\mathcal{I}}v_j/n)^2]/n$, for some $\gamma>0$, corresponds to the average processing speed minus a penalty that is proportional to the speed variance among threads.
\end{enumerate}

Any solution to the optimization problem (\ref{eq:CentralizedObjective}) would correspond to an \emph{efficient assignment}. Figure~\ref{fig:StaticApproach} presents a schematic of a \emph{static} resource allocation framework sequence of actions where the centralized objective (\ref{eq:CentralizedObjective}) is solved by the \RM\ once and then it communicates the optimal assignment to the threads.

However, there are two significant issues when posing an optimization problem in the form of (\ref{eq:CentralizedObjective}). In particular,
\begin{enumerate}
\item the function $v_i(\alpha,w)$ is unknown and it may only be evaluated through measurements of the processing speed, denoted $\tilde{v}_i$;
\item the exogenous influence $w$ is unknown and may vary with time, thus the optimal assignment may not be fixed with time.
\end{enumerate}

In conclusion, the static resource allocation framework of Figure~\ref{fig:StaticApproach} presented in (\ref{eq:CentralizedObjective}) is not easily implementable.

\subsection{Measurement- or learning-based optimization} 	\label{sec:MeasurementBasedOptimization}

We wish to target a \emph{static} objective of the form (\ref{eq:CentralizedObjective}) through a \emph{measurement-based} (or \emph{learning-based}) optimization approach. According to such approach, the \RM\ reacts to measurements of the objective function $f(\alpha,w)$, periodically collected at time instances $k=1,2,...$ and denoted $\tilde{f}(k)$. In the case of objective (O1), $\tilde{f}(k)\df\sum_{i=1}^{n}\tilde{v}_i(k)/n$. Given these measurements and the current assignment $\alpha(k)$ of resources, the \RM\ selects the next assignment of resources $\alpha(k+1)$ so that the measured objective approaches the true optimum of the unknown function $f(\alpha,w)$. In other words, the \RM\ employs an update rule of the form:
\begin{equation}	\label{eq:GenericUpdateRule}
\{(\tilde{v}_i(1),\alpha_i(1)),...,(\tilde{v}_i(k),\alpha_i(k))\}_i\mapsto\{\alpha_i(k+1)\}_i
\end{equation}
according to which prior pairs of measurements and assignments for each thread $i$ are mapped into a new assignment $\alpha_i(k+1)$ that will be employed during the next evaluation interval.

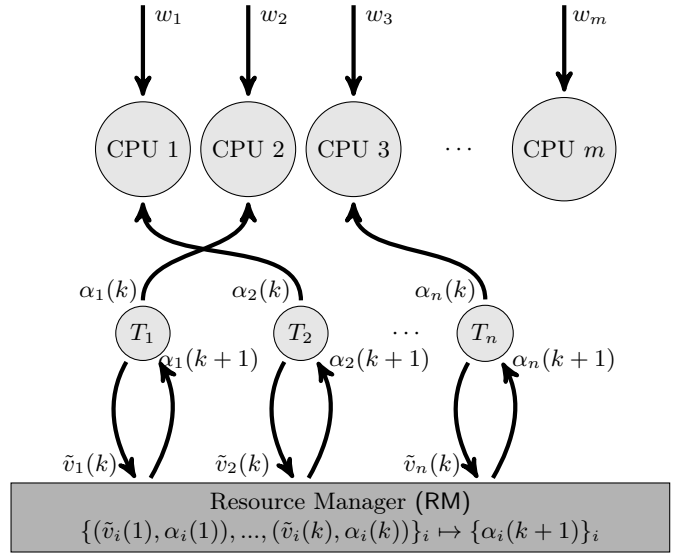
\begin{figure}[th!]
\centering
\begin{tikzpicture}[scale=3.5,
    axis/.style={thick, ->, >=stealth'},
    important line/.style={ultra thick},
    thin line/.style={thin},
    dashed line/.style={dashed, thin},
    arrow/.style={ultra thick, ->, >=stealth', shorten <=2pt, shorten >=2pt},
    every node/.style={color=black}
    ]
    \node[draw, rectangle, fill=LightGray] (RM) at (0.75,0) {\begin{minipage}{0.7\textwidth}\centering Resource Manager \textsf{(RM)} \\ 
    $\{(\tilde{v}_i(1),\alpha_i(1)),...,(\tilde{v}_i(k),\alpha_i(k))\}_i\mapsto\{\alpha_i(k+1)\}_i$\end{minipage}};
   
    \node[draw, circle,fill=VeryLightGray] (T1) at (0.0,.7) {$T_1$};
    \node[draw, circle,fill=VeryLightGray] (T2) at (0.6,.7) {$T_2$};
    \node[draw, circle,fill=VeryLightGray] (T3) at (1.3,0.7) {$T_n$};
    \node (dots) at (1,0.7) {$\mathbf{\cdots}$};
    
    \node[draw, circle, fill=VeryLightGray] (CPU1) at (0,1.4) {CPU $1$};
    \node[draw, circle, fill=VeryLightGray] (CPU2) at (0.4,1.4) {CPU $2$};
    \node[draw, circle, fill=VeryLightGray] (CPU3) at (0.8,1.4) {CPU $3$};
    \node (dots) at (1.2,1.4) {$\mathbf{\cdots}$};
    \node[draw, circle, fill=VeryLightGray] (CPUm) at (1.6,1.4) {CPU $m$};
    
    \node[] (RMT1) at (.0,.1) {};
    \node[] (RMT2) at (.6,.1) {};
    \node[] (RMT3) at (1.3,.1) {};
    
    \draw[arrow] (T1) to [out=-120,in=120] (RMT1);\node at (-0.2,0.2) {{\small $\tilde{v}_1(k)$}};
    \draw[arrow] (RMT1) to [out=60,in=-60] (T1);\node at (0.25,0.6) {{\small $\alpha_1(k+1)$}};
    \draw[arrow] (T2) to [out=-120,in=120] (RMT2);\node at (0.38,0.2) {{\small $\tilde{v}_2(k)$}};
    \draw[arrow] (RMT2) to [out=60,in=-60] (T2);\node at (0.9,0.6) {{\small $\alpha_2(k+1)$}};
    \draw[arrow] (T3) to [out=-120,in=120] (RMT3);\node at (1.1,0.2) {{\small $\tilde{v}_n(k)$}};
    \draw[arrow] (RMT3) to [out=60,in=-60] (T3);\node at (1.6,0.6) {{\small $\alpha_n(k+1)$}};
    
    \node at (-0.13,0.86) {{\small $\alpha_1(k)$}};
    \node at (0.45,0.86) {{\small $\alpha_2(k)$}};
    \node at (1.15,0.86) {{\small $\alpha_n(k)$}};
    
    \node[] (w1) at (.0,2) {};
    \draw[arrow] (w1) to [out=-90,in=90] (CPU1);\node at (0.1,1.9) {$w_1$};
    \node[] (w2) at (.4,2) {};
    \draw[arrow] (w2) to [out=-90,in=90] (CPU2);\node at (0.5,1.9) {$w_2$};
    \node[] (w3) at (.8,2) {};
    \draw[arrow] (w3) to [out=-90,in=90] (CPU3);\node at (0.9,1.9) {$w_3$};
    \node[] (wm) at (1.6,2) {};
    \draw[arrow] (wm) to [out=-90,in=90] (CPUm);\node at (1.7,1.9) {$w_m$};
    
    
    \draw[arrow] (T1) to [out=90,in=-90] (CPU2);
    \draw[arrow] (T2) to [out=90,in=-90] (CPU1);
    \draw[arrow] (T3) to [out=90,in=-90] (CPU3);

\end{tikzpicture}
\caption{Schematic of \emph{dynamic} resource allocation framework.}
\label{fig:DynamicApproach}
\end{figure}

The overall framework is illustrated in Figure~\ref{fig:DynamicApproach} describing the flow of information and steps executed. In particular, at any given time instance $k=1,2,...$, each thread $i$ communicates to the \RM\ its current processing speed $\tilde{v}_i(k)$. Then the \RM\ updates the assignments for each thread $i$, $\alpha_i(k+1)$, and communicates this assignment to them.


\subsection{Distributed learning}	\label{sec:DistributedOptimization}

Parallelized applications consist of multiple threads that can be controlled independently with respect to their CPU affinity (at least in Linux machines). Recently developed performance-recording tools (e.g., PAPI \cite{Mucci99}) also allows for a real-time collection of performance counters during the execution time of a thread. Thus, decisions over the assignment of CPU affinities can be performed independently for each thread, allowing for the introduction of a \emph{distributed learning} framework. Under such a framework, the \RM\ treats each thread as an independent decision maker and provides each thread with an independent decision rule. 

A distributed learning approach (i) \emph{reduces computation complexity}, since each thread has only $m$ available choices (instead of $m^{N}$ available group choices), and (ii) \emph{allows for an immediate response to changes observed in the environment} (e.g., available computing bandwidth), thus increasing the adaptivity and robustness of the resource allocation mechanism. For example, if the load of the $j$th CPU core increases (captured through the exogenous parameters $w_j$), then the thread(s) currently running on CPU $j$ may immediately react to this change without necessarily altering the assignment of the remaining threads. 

Such immediate reaction to changes in the environment constitutes a great advantage in comparison to centralized schemes. In the absence of an explicit form of the centralized objective (\ref{eq:CentralizedObjective}), a centralized framework would require a testing period over which all possible assignments are tested over an evaluation period and then compared with respect to their performance. When all possible assignments have been tested and evaluated, then the best one can be selected. Even if such optimization is repeated often, it is obvious that such \emph{exhaustive-search} approach will suffer from slow responses to possible changes in the environment.

It is also evident that the large evaluation period required by an exhaustive-search framework cannot provide any performance guarantees during the evaluation phase. In particular, since all alternative assignments need to be tested over some evaluation interval, bad assignments also have to be tried and evaluated. This may have an unpredictable impact in the overall performance of the application, thus reducing the impact of the optimization itself.

On the other hand, distributed learning schemes can be designed to allow only for small variations in the current assignment. For example, threads may experiment independently for alternative CPU assignments, however the frequency of such experimentations can be tuned independently for each thread. At the same time, an experimentation that leads to a worse assignment may always be reversed by the thread performing this experimentation, thus maintaining good performance throughout the execution time. Hence, distributed learning may allow for (iii)  \emph{a more direct and careful experimentation of the alternative options}. 

At the same time, distributed learning schemes can be designed to (iv) \emph{gradually approach at least locally optimum assignments}, which include all solutions to the static centralized optimization (\ref{eq:CentralizedObjective}). Thus, such schemes may provide guarantees over the performance of the overall parallelized application.

Points (i)--(iv) discussed above constitute the main advantages of a distributed learning scheme compared to a centralized approach.

\subsection{Objective}	\label{sec:Objective}

The objective in this paper is to address the problem of adaptive or dynamic pinning through a distributed learning framework. Each thread will constitute an independent decision maker or agent, thus naturally introducing a multi-agent formulation. Each thread selects its own CPU assignments independently using its own preference criterion (although the necessary computations for such selection are executed by the \RM). 

The goal is to design a preference criterion and a selection rule for each thread, so that when each thread tries to maximize its own (\emph{local}) criterion then certain guarantees can be achieved regarding the overall (\emph{global}) performance of the parallelized application. Furthermore, the selection criterion for each thread should be adaptive and robust to possible changes observed in the environment (e.g., the resource availability).

In the following sections, we will go through the design for such a distributed scheme, and we will provide guarantees with respect to its asymptotic behavior.

\section{Multi-Agent Formulation}	\label{sec:MultiagentFormulation}

The first step towards a distributed learning scheme is the decomposition of the decision making process into multiple decision makers (or agents). Naturally, in the problem of placing threads of a parallelized application into a set of available processing units, a thread may naturally constitute an independent decision maker. 

\subsection{Strategy}	\label{sec:Strategy}

Since each agent (or thread) selects actions independently, we generally assume that each agent's action is a realization of an independent discrete random variable. Let $\sigma_{ij}\in[0,1]$, $j\in\mathcal{A}_i$, denote the probability that agent $i$ selects its $j$th action in $\mathcal{A}_i$. If $\sum_{j=1}^{\magn{\mathcal{A}_i}}\sigma_{ij}=1$, then $\sigma_i\doteq(\sigma_{i1},...,\sigma_{i\magn{\mathcal{A}_i}})$ is a probability distribution over the set of actions $\mathcal{A}_i$ (or \emph{strategy} of agent $i$). Then $\sigma_i\in\SIMPLEX{\magn{\mathcal{A}_i}}$. To provide an example, consider the case of 3 available CPU cores, i.e., $\mathcal{A}_i=\{1,2,3\}$. In this case, the strategy $\sigma_i\in\SIMPLEX{3}$ of thread $i$ may take the following form:
\begin{equation*}
\sigma_{i} = \left(\begin{array}{c}
0.2 \\
0.5 \\
0.3
\end{array}\right),
\end{equation*}
such that $20\%$ corresponds to the probability of assigning itself to CPU core $1$, $50\%$ corresponds to the probability of assigning itself to CPU core $2$ and $30\%$ corresponds to the probability of assigning itself to CPU core $3$. Briefly, the assignment selection will be denoted by $$\alpha_i = \RAND{\sigma_i}{\mathcal{A}_i}.$$

We will also use the term \emph{strategy profile} to denote the combination of strategies of all agents $\sigma=(\sigma_1,...,\sigma_n)\in\SPROFILE$ where $\SPROFILE\doteq\SIMPLEX{\magn{\mathcal{A}_1}}\times ... \times \SIMPLEX{\magn{\mathcal{A}_n}}$ is the set of strategy profiles.

Note that if $\sigma_i$ is a unit vector (or a vertex of $\SIMPLEX{\magn{\mathcal{A}_i}}$), say $e_j$, then agent $i$ selects its $j$th action with probability one. Such a strategy will be called \emph{pure strategy}. Likewise, a \emph{pure strategy profile} is a profile of pure strategies. We will also use the term \emph{mixed strategy} to denote a strategy that is not pure.

\subsection{Utility function \& expected payoff}	\label{sec:UtilityFunction}

A cornerstone in the design of any measurement-based algorithm is the \emph{preference criterion} or \emph{utility function} $u_i$ for each thread $i\in\mathcal{A}$. The utility function captures the benefit of a decision maker (thread) as resulting from the assignment profile $\alpha$ selected by all threads, i.e., it represents a function of the form $u_i:\mathcal{A}\to\mathbb{R}$. Often, we may decompose the argument of the utility function as follows $u_i(\alpha) = u_i(\alpha_i,\alpha_{-i})$, where $-i\doteq\mathcal{I}\backslash{i}$. The utility function introduces a preference relation for each decision maker where $u_i(\alpha_i,\alpha_{-i}) \geq u_{i}(\alpha_i',\alpha_{-i})$ translates to $\alpha_i$ being more desirable/preferable than $\alpha_i'$.

It is important to note that the utility function $u_i$ of each agent/thread $i$ is subject to \emph{design} and it is introduced in order to guide the preferences of each agent. Thus, $u_i$ may not necessarily correspond to a measured quantity, but it could be a function of available performance counters. 

For example, a natural choice for the utility of each thread is its own execution speed $v_i$. Other options may include more egalitarian criteria, where the utility function of each thread corresponds to the overall global objective $f(\alpha,w)$. The definition of a utility function is open-ended.

\subsection{Assignment Game}		\label{sec:AssignmentGame}

Assuming that each thread (or agent) may decide independently on its own CPU placement, so that its preference criterion is maximized, a strategic-form game between the running threads can naturally be introduced. We define it as a strategic interaction or game because the strategy of each thread indirectly influences the performance of the other threads, thus introducing an interdependence between their utility functions. We define the triple $\{\mathcal{I},\mathcal{A},\{u_i\}_i\}$ as an \emph{assignment game}.

\subsection{Nash Equilibria}		\label{sec:NashEquilibria}

Given a strategy profile $\sigma\in\SPROFILE$, the \emph{expected payoff vector} of each agent $i$, $U_i:\SPROFILE\to\mathbb{R}^{\magn{\mathcal{A}_i}}$, can be computed by
\begin{equation}	\label{eq:ExpectedUtility}
U_i(\sigma) \doteq \sum_{\alpha_i\in\mathcal{A}_i}e_{\alpha_i}\sum_{\alpha_{-i}\in\mathcal{A}_{-i}}\left(\prod_{s\in{-i}}\sigma_{s\alpha_{s}}\right)u_i(\alpha_i,\alpha_{-i}).
\end{equation}
We may think of the $j$th entry of the expected payoff vector $U_i$, denoted $U_{ij}(\sigma)$, as the expected payoff of agent $i$ playing action $j$ at strategy profile $\sigma$. 
Finally, let $u_i(\sigma)$ be the \emph{expected payoff} of agent $i$ at strategy profile $\sigma\in\SPROFILE$, which satisfies:
\begin{equation}
u_i(\sigma) = \sigma_i\tr U_i(\sigma).
\end{equation}

\begin{definition}[Nash Equilibrium]	\label{eq:NashEquilibrium}
A strategy profile $\sigma^*=(\sigma_1^*,...,\sigma_n^*)\in\SPROFILE$ is a Nash equilibrium if, for each agent $i\in\mathcal{I}$, 
\begin{equation}	\label{eq:NashCondition}
u_i(\sigma_i^*,\sigma_{-i}^*) \geq u_i(\sigma_i,\sigma_{-i}^*),
\end{equation}
for all $\sigma_i\in\SIMPLEX{\magn{\mathcal{A}_i}}$ with $\sigma_i\neq\sigma_i^*$.
\end{definition}
In other words, a strategy profile is a Nash equilibrium when no agent has the incentive to change this strategy (given that every other agent does not change its strategy). In the special case where for all $i\in\mathcal{I}$, $\sigma_i^*$ is a pure strategy, then the Nash equilibrium is called \emph{pure Nash equilibrium}.

\subsection{Efficient assignments vs Nash equilibria}	\label{sec:EfficientAllocationsVSNashEquilibria}

As we shall see in a forthcoming section, Nash equilibria can be potential attractors of several classes of distributed learning schemes, therefore their relation to the efficient assignments becomes important. 

Nash equilibria correspond to \emph{locally} stable equilibria (with respect to the agents' preferences), i.e., no agent has the incentive to alter its strategy. On the other hand, \emph{efficient assignments} correspond to strategy profiles that maximize the global objective (\ref{eq:CentralizedObjective}). As probably expected, a Nash equilibrium does not necessarily coincide with an efficient assignment and vice versa. Both the utility function of each agent $i$, $u_i$, as well as the global objective $f(\alpha,w)$ are \emph{subject to design}, and their selection determines the relation between Nash equilibria and efficient assignments. 

The \RM\ can be designed to have access to the performances of all threads. Thus, a natural choice for the utility of each thread can be the overall objective function, i.e., 
\begin{equation}	\label{eq:AssignmentGameUtility}
u_i(\alpha) \doteq f(\alpha,w),
\end{equation}
for some given exogenous factor $w$. Note that this definition is independent of whether objective (O1) or (O2) is selected. Such classes of strategic interactions where the utilities of all independent agents are identical, are referred to as \emph{identical interest games} and they are part of a larger family of games, namely \emph{potential games}. It is straightforward to check that in this case, \emph{the set of efficient assignments belongs to the set of Nash equilibria (locally optimal allocations)}. In this case, it is desirable that agents learn to select placements that correspond to Nash equilibria, since a) it provides a minimum performance guarantee (since \emph{all} non-locally optimal placements are excluded), and b) it increases the probability for converging to the solution(s) of the global objective (\ref{eq:CentralizedObjective}).


\section{Reinforcement Learning (RL)}	\label{sec:ReinforcementLearning}

In the previous section, we introduced utility functions for each thread (or agent), so that the set of efficient assignments (\ref{eq:CentralizedObjective}) are restricted within the set of Nash equilibria. However, as we have already discussed in Section~\ref{sec:MeasurementBasedOptimization}, the utility function of each thread is not known a-priori, rather it may only be measured after the selection of a particular assignment is in place. Thus, the question that naturally arises is \emph{how agents may choose assignments based only on their available measurements so that eventually an efficient assignment is established for all threads}.

To this end, we employ a distributed learning framework (namely, \emph{perturbed learning automata}) that is based on the reinforcement learning algorithm introduced in \cite{ChasparisShamma11_DGA,ChasparisShammaRantzer14}. It belongs to the general class of \emph{learning automata} \cite{Narendra89}.

The basic idea behind reinforcement learning is rather simple. If agent $i$ selects action $j$ at instance $k$ and a favorable payoff results, $u_i(\alpha)$, the action probability $\sigma_{ij}(k)$ is increased and all other entries of $\sigma_i(k)$ are decreased.

The precise manner in which $\sigma_i(k)$ changes, depending on the assignment $\alpha_i(k)$ performed at stage $k$ and the response $u_i(\alpha(k))$ of the environment, completely defines the reinforcement learning model.

\subsection{Strategy update}	\label{sec:StrategyUpdate}

According to the \emph{perturbed reinforcement learning} \cite{ChasparisShamma11_DGA,ChasparisShammaRantzer14}, the strategy of each thread at any time instance $k=1,2,...$ is as follows:
\begin{equation}	\label{eq:SelectionRule}
\sigma_i(k) = (1-\lambda)x_i(k) - \frac{\lambda}{\magn{\mathcal{A}_i}}
\end{equation}
where $\lambda>0$ corresponds to a perturbation term (or \emph{mutation}) and $x_i(k)$ corresponds to the \emph{nominal strategy} of agent $i$. The nominal strategy is updated according to the following update recursion:
\begin{equation}	\label{eq:StrategyUpdate}
x_i(k+1) = \Pi_{\SIMPLEX{\magn{\mathcal{A}_i}}}\left[x_i(k) + \epsilon u_i(\alpha(k))[e_{\alpha_i(k)}-x_i(k)]\right],
\end{equation}
for some constant step-size $\epsilon>0$. Note that according to this recursion, the new nominal strategy will increase in the direction of the action $\alpha_i(k)$ which is currently selected and it will increase proportionally to the utility received from this selection. For sufficiently small step size $\epsilon>0$ and given that the utility function $u_i(\cdot)$ is uniformly bounded for all action profiles $\alpha\in\mathcal{A}$, the projection operator $\Pi_{\SIMPLEX{\magn{\mathcal{A}_i}}}[\cdot]$ can be skipped.

In comparison to \cite{ChasparisShamma11_DGA,ChasparisShammaRantzer14}, the difference here lies in the use of the constant step size $\epsilon>0$ (instead of a decreasing step-size sequence). This selection increases the adaptivity and robustness of the algorithm to possible changes in the environment. This is because a constant step size provides a fast transition of the nominal strategy from one pure strategy to another. 



Furthermore, the reason for introducing the perturbation term $\lambda$ is to provide the possibility for the nominal strategy to escape from pure strategy profiles, that is profiles at which all agents assign probability one in one of the actions. Setting $\lambda>0$ is essential for providing an adaptive response of the algorithm to changes in the environment.

\section{Convergence Analysis}		\label{sec:ConvergenceAnalysis}

In this section, we establish a connection between the asymptotic behavior of the nominal strategy profile $x(k)$ with the Nash equilibria of the assignment game, when the utility function $u_i$ for each thread $i$ is defined by (\ref{eq:AssignmentGameUtility}) and the objective is  given by either (O1) or (O2). Let us denote $\mathcal{S}^{\lambda}$ to be the set of \emph{stationary points} of the mean-field dynamics (cf.,~\cite{KushnerYin03}) of the recursion (\ref{eq:StrategyUpdate}) (when the projection operator has been skipped), defined as follows
\begin{equation*}	\label{eq:StationaryPoints}
\mathcal{S}^{\lambda}\df \left\{ x\in\SPROFILE : g_i^{\lambda}(x)\df\mathbb{E}\left[u_i(\alpha(k))[e_{\alpha_i(k)} - x_i(k)]|x(k)=x\right] = 0, \forall i\in\mathcal{I} \right\}.
\end{equation*}
The expectation operator $\mathbb{E}[\cdot]$ is defined appropriately over the canonical path space $\Omega=\SPROFILE^{\infty}$ with an element $\omega$ being a sequence $\{x(0),x(1),...\}$ with $x(k)=(x_1(k),...,x_{n}(k))\in\SPROFILE$  generated by the reinforcement learning process. Similarly we define the probability operator $\mathbb{P}[\cdot]$. In other words, the set of stationary points corresponds to the strategy profiles at which the expected change in the strategy profile is zero.

According to \cite{ChasparisShamma11_DGA,ChasparisShammaRantzer14}, a connection can be established between the set of stationary points $\mathcal{S}^{\lambda}$ and the set of Nash equilibria of the assignment game. In particular, for sufficiently small $\lambda>0$, \emph{the set of $\mathcal{S}^{\lambda}$ includes only $\lambda$-perturbations of Nash-equilibrium strategies}. This is due to the fact that the mean-field dynamics $\{g_i^{\lambda}(\cdot)\}_i$ are continuously differentiable functions with respect to $\lambda$.\footnote{For more information regarding the sensitivity of stationary points to $\lambda>0$, see \cite[Lemma~6.2]{ChasparisShammaRantzer14}.}

The following proposition is a straightforward extension of \cite[Theorem~1]{ChasparisShammaRantzer14} to the case of constant step-size.
\begin{proposition}	\label{Pr:InfinitelyOftenVisits}
Let the \RM\ employ the strategy update rule (\ref{eq:StrategyUpdate}) and placement selection (\ref{eq:SelectionRule}) for each thread $i$. Updates are performed periodically with a fixed period such that $\tilde{v}_i(k)>0$ for all $i$ and $k$. Let the utility function for each thread $i$ satisfy (\ref{eq:AssignmentGameUtility}) under either objective (O1) or (O2), where $\gamma\geq{0}$ is small enough such that $u_i(\alpha(k))>{0}$ for all $k$. 

Then, for some $\lambda>0$ sufficiently small, there exists $\delta=\delta(\lambda)$, with $\delta(\lambda)\downarrow{0}$ as $\lambda\downarrow{0}$, such that
\begin{equation}	\label{eq:Convergence}
\mathbb{P}\left[\liminf_{k\to\infty} \mathsf{dist}(x(k),\mathcal{B}_{\delta}(\mathcal{S}^{\lambda}))=0\right]=1.
\end{equation}
\end{proposition}
\begin{proof}
The proof follows the exact same steps of the first part of \cite[Theorem~1]{ChasparisShammaRantzer14}, where the decreasing step-size sequence is being replaced by a constant $\epsilon>0$.
\end{proof}
Proposition~\ref{Pr:InfinitelyOftenVisits} states that when we select $\lambda$ sufficiently small, the nominal strategy trajectory will be approaching the set $\mathcal{B}_{\delta}(\mathcal{S}^{\lambda})$ infinitely often with probability one, that is a small neighborhood of the Nash equilibria. We require that the update period is large enough so that each thread is using resources within each evaluation period. Of course, if a thread stops executing then the same result holds but for the updated set of threads. 

However, the above proposition does not provide any guarantees regarding the time fraction that the process spends in any Nash equilibrium. The following proposition establishes this connection. 

\begin{proposition}[Weak convergence to Nash equilibria]	\label{Pr:WeakConvergence}
\textit{
Under the hypotheses of Proposition~\ref{Pr:InfinitelyOftenVisits}, the fraction of time that the nominal strategy profile $x(k)$ spends in $\mathcal{B}_{\delta}(\mathcal{S}^{\lambda})$ goes to one (in probability) as $\epsilon\to{0}$ and $k\to\infty$.
}
\end{proposition}
\begin{proof}
The proof follows directly from \cite[Theorem~8.4.1]{KushnerYin03} and Proposition~\ref{Pr:InfinitelyOftenVisits}.
\end{proof}

Proposition~\ref{Pr:WeakConvergence} states that if we take a small step size $\epsilon>0$, then as the time index $k$ increases, we should expect that the nominal strategy spends the majority of the time within a small neighborhood of the Nash equilibrium strategies. According to Section~\ref{sec:EfficientAllocationsVSNashEquilibria}, we know that when the utility function for each thread is defined according to (\ref{eq:AssignmentGameUtility}), then \emph{the set of Nash equilibria includes the set of efficient assignments}, i.e., the solutions of (\ref{eq:CentralizedObjective}). Thus, due to Proposition~\ref{Pr:WeakConvergence}, it is guaranteed that the nominal strategies $x_i(k)$, $i\in\mathcal{I}$, will spend the majority of the time in a small neighborhood of locally-optimal assignments, which provides a minimum performance guarantee throughout the running time of the parallelized application. 

Note that due to varying exogenous factors, the Nash-equilibrium assignments may not stay fixed for all future times. The above proposition states that the process will spend the majority of the time within the set of the Nash-equilibrium assignments for as long as this set is fixed. If, at some point in time, this set changes (due to, e.g., other applications start running on the same platform), then the above result continues to hold but for the new set of Nash equilibria. Hence, the process is adaptive to possible performance variations.

\section{Experiments}		\label{sec:Experiments}

In this section, we present an experimental study of the proposed reinforcement learning scheme for dynamic pinning of parallelized applications. Experiments were conducted on \texttt{20$\times$Intel\copyright Xeon\copyright CPU E5-2650 v3 \@ 2.30 GHz} running Linux Kernel 64bit 3.13.0-43-generic. The machine divides the physical cores into two NUMA nodes (Node 1: 0-9 CPU's, Node 2: 10-19 CPU's).

\subsection{Experimental Setup}

We consider a computationally intensive routine that executes a fixed number of computations (corresponding to the combinations of $M$ out of a set of $N>M$ numbers). The routine is being parallelized using the \texttt{pthread.h} (C++ POSIX thread library), where each thread is executing a replicate of the above set of computations. The nature of these computations does not play any role and in fact it may vary between threads (as we shall see in both of the forthcoming experiments).

Throughout the execution, and with a fixed period of $0.3$ sec, the \RM\ collects measurements of the total instructions per sec (using the PAPI library \cite{Mucci99}) for each one of the threads separately. Given the provided measurements, the update rule of Equation~(\ref{eq:StrategyUpdate}) with the utility function (\ref{eq:AssignmentGameUtility}) under (O2) is executed by the \RM. Placement of the threads to the available CPU's is achieved through the \texttt{sched.h} library (in particular, the \texttt{pthread\_setaffinity\_np} function). In the following, we demonstrate the response of the RL scheme in comparison to the Operating System (OS) response (i.e., when placement of the threads is not controlled by the \RM). We compare them for different values of $\gamma\geq{0}$ in order to investigate the influence of more balanced speeds to the overall running time.

In all the forthcoming experiments, the \RM\ is executed within the master thread which is always running in the first available CPU (CPU~1). Furthermore, in all experiments, only the first one of the two NUMA nodes are utilized, since our intention is to demonstrate the potential benefit of an efficient thread placement when the effect of memory placement is rather small.

\begin{figure}[h!]
\centering
\begin{minipage}{0.99\textwidth}
\includegraphics{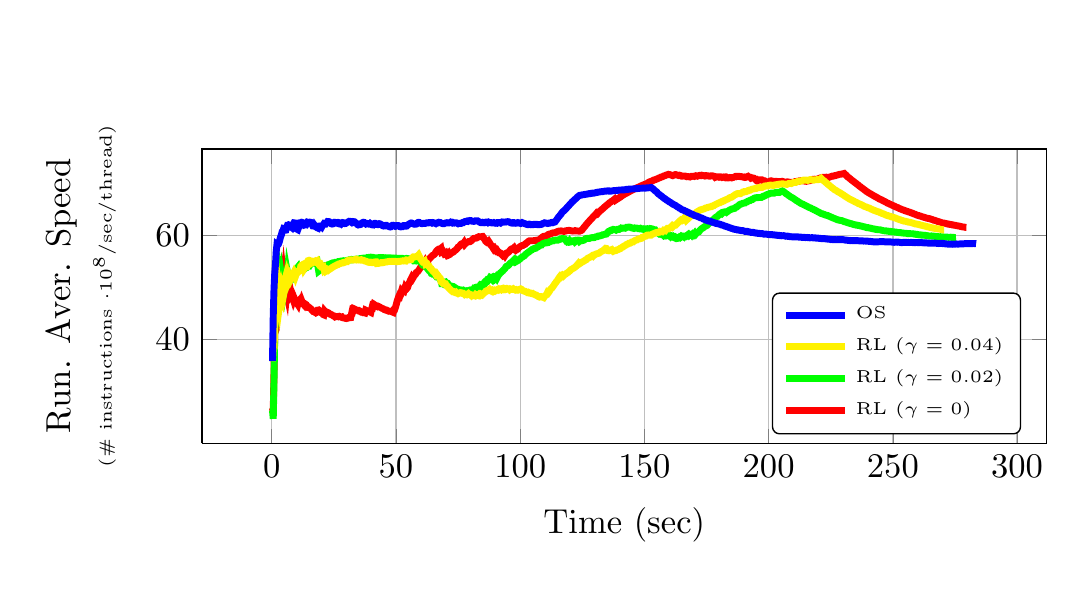} \\[-20pt]
\includegraphics{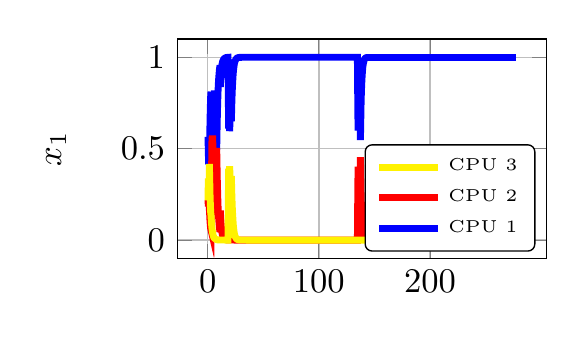}
\includegraphics{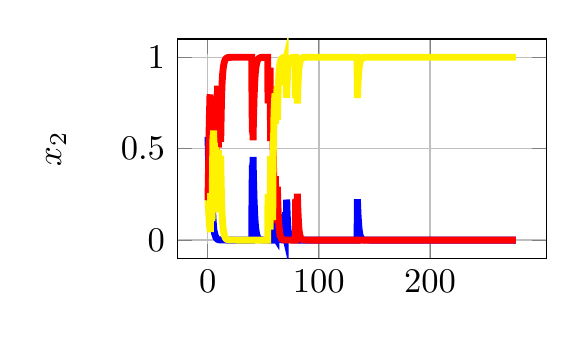} \\[-20pt]
\includegraphics{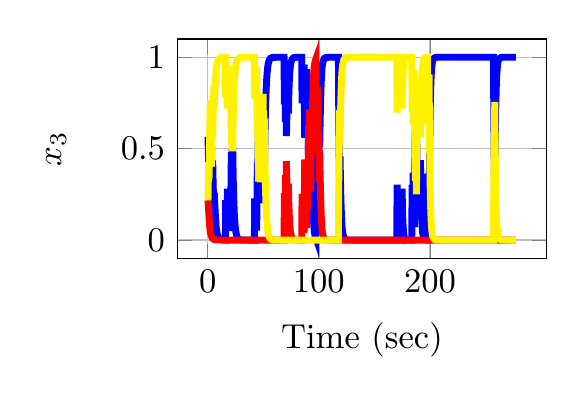}
\includegraphics{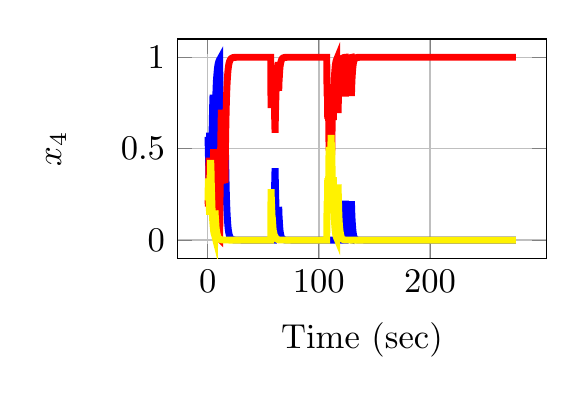}
\end{minipage}
\caption{Running average execution speed when 4 threads run on 3 identical CPU's. Thread 3 requires about half the computing bandwidth compared to the rest of the threads which are identical. The strategies correspond to the RL scheme with $\gamma=0.04$. The RL schemes run with $\epsilon=0.005$ and $\lambda=0.005$.}
\label{fig:Experiment3}
\end{figure}

\subsection{Experiment 1: Non-Identical Threads under Limited Resources}	\label{sec:Experiment3}

In this experiment, we consider the case of limited resources (i.e., when the number of threads is larger than the number of available CPU's). However, one of the threads requires CPU time with smaller frequency than the rest of the threads (i.e., executes its computations with smaller frequency). We should expect that in an optimal setup, threads that do not require CPU time often should be the ones sharing a CPU. On the other hand, threads that require larger bandwidth, they should be placed alone. 

In particular, in this experiment, Thread 3 requires about half the computing bandwidth compared to the rest of the threads (Thread 1, 2 and 4). The resulting performance is depicted in Figure~\ref{fig:Experiment3}.

We observe indeed that Threads 1, 2 \& 4 (which require larger computing bandwidths) are allocated to different CPU's (CPU 1, 3 and 2, respectively). On the other hand, Thread 3 is switching between CPU 1 and CPU 3, since both provide almost equal processing bandwidth to Thread 3. In other words, the less demanding application is sharing the CPU with one of the more demanding threads. Note that this assignment corresponds to a Nash equilibrium (as Proposition~\ref{Pr:WeakConvergence} states), since there is no thread that can benefit by changing its strategy. It is also straightforward to check that this assignment is also efficient.

Note, finally, that the difference with the processing speed of the OS scheme is small, although a more balanced processing speed ($\gamma=0.04$) improved slightly the overall completion time.

%
%

\begin{figure}[t!]
\centering
\begin{minipage}{0.99\textwidth}
\includegraphics{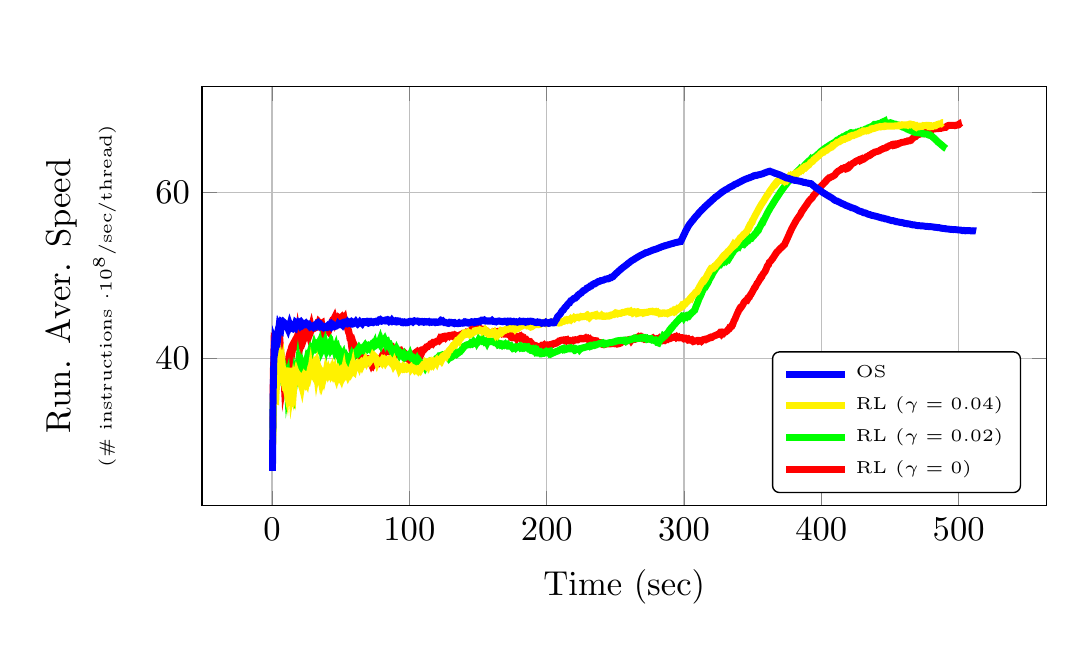}
\end{minipage}
\caption{Running average execution speed when 7 non-identical threads run on 3 CPU cores. Threads~1 \& 2 require about half the computing bandwidth compared to the rest of the threads (which are identical). Thread 3 is joining after 120 sec. The RL schemes run with $\epsilon=0.003$ and $\lambda=0.005$.}
\label{fig:Experiment5}
\end{figure}

\subsection{Experiment 2: Non-Identical Threads in a Dynamic Environment}		\label{sec:Experiment4}

In this experiment, we demonstrate the robustness of the algorithm in a dynamic environment. We consider 7 threads. The first two (Threads 1 \& 2) require about half the computing bandwidth compared to the rest. The rest of the threads (Threads 3, 4, 5, 6 and 7) are identical. However, Thread 3 starts running later in time (in particular, after 120 sec). 

Figure~\ref{fig:Experiment5} illustrates the evolution of the RL scheduling scheme under different values of $\gamma$. Again in this case, a fastest response of the overall application can be achieved when higher values of $\gamma$ are selected. The difference should be attributed to the fact that the OS fails to distinguish between threads with different bandwidth requirements. Table~\ref{Tb:ComparisonOSwithRL} presents a statistical analysis of these schemes where the speed difference between the RL ($\gamma=0.04$) and the OS reaches approximately 5\% on average.

\begin{table}[h!]
\centering
\begin{tabular}{c|c|c|c|c}
Run \#			& 	OS 					&  RL ($\gamma=0$) 	& 	RL ($\gamma=0.02$)	 & RL ($\gamma=0.04$)			\\\hline\hline
1 				&  513 sec				& 	505 sec			&  492 sec 				 & 489 sec						\\\hline
2				&  530 sec				& 	506 sec			&  489 sec				 & 494 sec						\\\hline
3 				&  536 sec  			& 	517	sec			&  518 sec				 & 515 sec 						\\\hline
4 				&  533 sec  			& 	507	sec			&  515 sec				 & 509 sec						\\\hline
5				&  523 sec				& 	502	sec			&  491 sec  			 & 496 sec						\\\hline
6				&  513 sec				& 	523	sec			&  501 sec				 & 492 sec						\\\hline
7				&  520 sec				& 	514	sec			&  497 sec				 & 492 sec						\\\hline
8				&  530 sec				& 	518	sec			&  499 sec				 & 497 sec						\\\hline
9				&  520 sec				& 	532	sec			&  500 sec				 & 497 sec						\\\hline
10				&  528 sec				& 	517	sec			&  493 sec				 & 492 sec						\\\hline\hline
\textbf{aver.}	&  \textbf{524.6 sec}	& 	\textbf{514.10}	&  \textbf{499.5 sec}	 & \textbf{497.3}				\\\hline\hline
\textbf{s.d.}	&  \textbf{8.06 sec}	& 	\textbf{9.29}	&  \textbf{9.85	sec}	 & \textbf{8.27}				\\\hline
\end{tabular}
\caption{Comparison between the OS performance and RL schemes when $\epsilon=0.003$ and $\lambda=0.005$ for different values of $\gamma$ under Experiment 2.}
\label{Tb:ComparisonOSwithRL}
\end{table}

\section{Conclusions} \label{sec:Conclusions}

We proposed a measurement-based learning scheme for addressing the problem of efficient dynamic pinning of parallelized applications into processing units. According to this scheme, a centralized objective is decomposed into thread-based objectives, where each thread is assigned its own utility function. A \RM\ updates a strategy for each one of the threads corresponding to its beliefs over the most beneficial CPU placement for this thread. Updates are based on a reinforcement learning rule, where prior actions are reinforced proportionally to the resulting utility. It was shown that, when we appropriately design the threads' utilities, then convergence to the set of locally optimal assignments is achieved. Besides its reduced computational complexity, the proposed scheme is adaptive and robust to possible changes in the environment.

\bibliographystyle{splncs03}
\bibliography{2016_HLPP_Bibliography}

\end{document}